\documentclass[submission,copyright,creativecommons]{eptcs}
\providecommand{\event}{TARK 2017} % Name of the event you are submitting to
 \usepackage{underscore}           % Only needed if you use pdflatex.

\usepackage{amsmath,stmaryrd, graphicx, amssymb}
\usepackage{amsthm}
\usepackage{breakcites}
\usepackage{tikz}

\graphicspath{ {} }

\newtheorem{proposition}{Proposition} [section]
\newtheorem{corollary}[proposition]{Corollary}

\newtheorem{lemma}[proposition]{Lemma}

\theoremstyle{definition}
\newtheorem {definition}[proposition]{Definition}

\title{Coalition and Group Announcement Logic\footnote{
This is a corrected version of \cite{galimullin17}. The previous version considered \textbf{CoGAL}, a combination of \textbf{CAL}
and \textbf{GAL} without relativised operators. There is a gap in the completeness proof of \textbf{CoGAL} given in 
\cite{galimullin17}. Specifically, the proof of the Lindenbaum Lemma (Proposition 2.19) fails to demonstrate that when adding a 
witness $\psi_G$ for $\neg \eta_i([ \! \langle G \rangle \! ] \varphi_i)$, we also have all the corresponding formulas with 
$\chi_{A \setminus G}$ (which is required by the semantics). Completeness of \textbf{CoGAL} is hence an open question. 
In this corrected version we consider relativised group announcement operators instead of \textbf{GAL} operators. This allows us 
to give a sound and complete axiomatisation of \textbf{CoRGAL}. We omit Propositions 2.11 and 2.12 of \cite{galimullin17} that 
also have errors in proofs.
We would like to acknowledge discussions with Hans van Ditmarsch and Tim French that helped us to identify and correct the errors.
} \\ (Corrected Version) }

\author{Rustam Galimullin \qquad\qquad Natasha Alechina
\institute{School of Computer Science\\
University of Nottingham\\
Nottingham, UK}
\email{\{rustam.galimullin, natasha.alechina\}@nottingham.ac.uk}
}
\def\titlerunning{Coalition and Group Announcement Logic}
\def\authorrunning{Rustam Galimullin \& Natasha Alechina}
\begin{document}
\maketitle

%\author{
%  Rustam Galimullin \qquad Natasha Alechina\\
%  University of Nottingham, UK\\
%  \texttt{\{rustam.galimullin, natasha.alechina\}@nottingham.ac.uk}
%}
%\begin{document}
%\maketitle

\begin{abstract}
Dynamic epistemic logics which model abilities of agents to make various announcements and influence each other's knowledge have been studied extensively in recent years. Two notable examples of such logics are Group Announcement Logic and Coalition Announcement Logic. They allow us to reason about what groups of agents can achieve through joint announcements in non-competitive and competitive environments. In this paper, we consider a combination of these logics -- Coalition and Relativised Group Announcement Logic and provide its complete axiomatisation. Moreover, we partially answer the question of how group and coalition announcement operators interact, and settle some other open problems.
\end{abstract}

\section{Introduction}

To introduce the logics we will be working with in this paper, we start with an example loosely based on the one from \cite{renne09}.
Let us imagine that Ann, Bob, and Cath are travelling by train from Nottingham to Liverpool through Manchester. Cath was sound asleep all the way, and she has just woken up. She does not know whether the train passed Manchester, but Ann and Bob know that it has not. Now, if the train driver announces that the train is approaching Manchester, then Cath, as well as Ann and Bob, knows that they have not passed the city yet. To reason about changes in agents' knowledge after public announcements, we can use Public Announcement Logic (\(\mathbf{PAL}\)) \cite{plaza07}. Returning to the example, let us assume that the train driver does not announce anything, so that Cath is not aware of her whereabouts. Ann and Bob may tell her whether they passed Manchester.
%, or they may say some generally known truth (e.g. that it is either raining in Liverpool or not). 
In other words, Ann and Bob have an announcement that can influence Cath's knowledge. An extension of \(\mathbf{PAL}\), Group Announcement Logic (\(\mathbf{GAL}\)) \cite{agotnes10}, deals with the \emph{existence} of announcements by groups of agents that can achieve certain results. Now, let us assume that Ann does not want to disclose to Cath their whereabouts and Bob does, i.e. Ann and Bob have different goals. Then, it is clear that no matter what Ann says, the coalition of Bob and Cath can achieve the goal of Cath knowing that the train has not passed Manchester, that is, Bob can communicate this information to Cath. On the other hand, if Ann and Bob work together, then they have an announcement (for example, a tautology `It either rains in Liverpool or it doesn't') such that whatever Cath says, she remains unaware of her whereabouts. For this type of strategic behaviour, another extension of \(\mathbf{PAL}\) -- Coalition Announcement Logic (\(\mathbf{CAL}\)) -- has been introduced in \cite{agotnes08}.

$\mathbf{CAL}$ joins two logical traditions: Dynamic Epistemic Logic, of which \(\mathbf{PAL}\) is a representative, and Coalition Logic (\(\mathbf{CL}\)) \cite{pauly02}. The latter allows us to reason about whether a coalition of agents has a strategy to achieve some goal, no matter what the agents outside of the coalition do. \(\mathbf{CL}\) essentially talks about concurrent games, and the actions that the agents execute are arbitrary actions (strategies in one-shot games). So, from this perspective, \(\mathbf{CAL}\) is a coalition logic with available actions restricted to public announcements. 

To the best of our knowledge, there is no complete axiomatisation of \(\mathbf{CAL}\) \cite{agotnes08, vanditmarsch12, agotnes14, agotnes16} or any other logic with coalition announcement operators. In this paper, we consider Coalition and Relativised Group Announcement Logic (\(\mathbf{CoRGAL}\)), a combination of an extension of \(\mathbf{GAL}\) and \(\mathbf{CAL}\), which includes operators for both group and coalition announcements. The main result of this paper is a sound and complete axiomatisation of  \(\mathbf{CoRGAL}\). As part of this result, we study the interplay between group and coalition announcement operators, and partially settle the question on their interaction that was stated as an open problem in \cite{vanditmarsch12, agotnes16}. 

\section{Coalition and Relativised Group Announcement Logic} \label{sec:CoRGAL}
\subsection {Syntax and Semantics}

Throughout the paper, let a finite set of agents \(A\), and a countable set of propositional variables \(P\) be given. The language of the logic is comprised of the language of classical propositional logic with added operators for agents' knowledge \(K_a \varphi\) (reads `agent \(a\) knows \(\varphi\)'), and public announcement \([\psi] \varphi\) (reads `after public announcement that \(\psi\), \(\varphi\) holds'), relativised group announcement \([G, \chi] \varphi\) (`given some announcement $\chi$, whatever agents from \(G\) announce at the same time, \(\varphi\) holds afterwards'), and coalition announcements \([ \! \langle G \rangle \! ] \varphi\) (`for every public announcement by coalition of agents \(G\) there is an announcement by other agents \(A \setminus G\), such that \(\varphi\) holds after joint simultaneous announcement').

\begin{definition} (Language)
    The \emph{language of coalition and} \emph{relativised group announcement} \emph{logic} \(\mathcal{L}_{CoRGAL}\) is as follows:
    \begin{center}
        \(\varphi,\psi ::= p \mid \neg \varphi \mid (\varphi \wedge \psi) \mid K_a \varphi \mid [\varphi]\psi \mid [G, \psi]\varphi \mid [ \! \langle G \rangle \! ] \varphi\),
    \end{center}
    where \(p \in P\), \(a \in A\), \(G \subseteq A\), and all the usual abbreviations of propositional logic (such as \(\vee, \rightarrow, \leftrightarrow\)) and conventions for deleting parentheses hold. The dual operators are defined as follows: \(\widehat {K}_a \varphi \leftrightarrow \neg K_a \neg \varphi\), \(\langle \varphi \rangle \psi \leftrightarrow \neg [\varphi] \neg \psi\), \(\langle G, \psi \rangle \varphi \leftrightarrow \neg [ G, \psi ] \neg \varphi\), and \( \langle \! [ G ] \! \rangle \varphi \leftrightarrow \neg [ \! \langle G \rangle \! ] \neg \varphi\).  
Observe that  $\langle G, \psi \rangle \varphi$ means that $G$ has an announcement 
such that after announcing it in conjunction with $\psi$, $\varphi$ holds, and $\langle \! [ G ] \! \rangle \varphi$ means 
that $G$ has an announcement such that after it is made simultaneously with 
any announcement by $A\setminus G$, $\varphi$ holds. The latter corresponds
to the Coalition Logic operator, but for announcements instead of arbitrary 
actions.

We define \(\mathcal{L}_{RGAL}\) as the language without the operator \([ \! \langle G \rangle \! ]\), \(\mathcal{L}_{PAL}\) the language without \([G,\psi]\) as well, and \(\mathcal{L}_{EL}\) the purely epistemic language which in addition does not contain announcement operators \([\varphi]\).

\end{definition}

Next definition is needed for technical reasons in the formulation of infinite rules of inference in Definition \ref{def::axiomatisation}. We want the rules to work for a class of different types of premises. Ultimately, we require premises to be expressions of  depth \(n\) of the type \(\varphi_1 \rightarrow \square_1 (\varphi_2 \rightarrow \mathellipsis (\varphi_n \rightarrow \square_n \sharp) \mathellipsis)\), where \(\square_i\) is either \(K_a\) or \([\psi]\) for some \(a \in A\) and \(\psi \in \mathcal{L}_{CoRGAL}\), atom \(\sharp\) denotes a placement of a formula to which a derivation is applied,  and some \(\varphi\)'s and \(\square\)'s can be omitted. This condition is captured succinctly by necessity forms originally introduced by Goldblatt in \cite{goldblatt}.

\begin{definition} (Necessity forms)
    Let \(\varphi \in \mathcal{L}_{CoRGAL}\), then \emph{necessity forms} \cite{goldblatt} are inductively defined as follows:
    \begin{center}
        \(\eta ::= \sharp \mid \varphi \rightarrow \eta (\sharp) \mid K_a \eta (\sharp) \mid [\varphi] \eta (\sharp)\).
    \end{center}
    The atom \(\sharp\) has a unique occurrence in each necessity form. The result of the replacement of \(\sharp\) with \(\varphi\) in some \(\eta (\sharp)\) is denoted as \(\eta (\varphi)\). 
\end{definition}

Whereas formulas of coalition logic \cite{pauly02} are interpreted in game structures, formulas of \(\mathbf{CoRGAL}\) are interpreted in epistemic models. Let us consider an example of such a model first. 
\begin{figure}[h]
    \centering
    \includegraphics {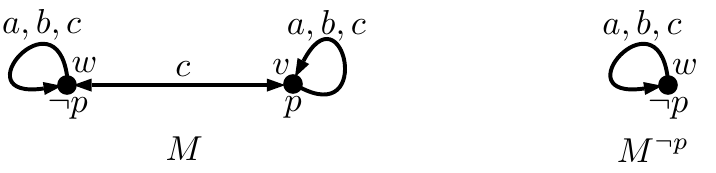}
    \caption{Train example}
    \label{fig::tark1}
\end{figure}
In Figure \ref{fig::tark1} there are three agents: \(a\) (Ann), \(b\) (Bob), and \(c\) (Cath). Let \(p\) denote the proposition that `The train has passed Manchester.' There are two states in the model \(M\): a state $w$ where
\(\neg p\) is true, and a state $v$ where \(p\) is true; and only one state in model \(M^{\neg p}\) which denotes $M$ updated by the announcement $\neg p$ (the process of updating the model is described below). % States represent possible worlds from agents' viewpoint. 
Let the $w$ be the actual state. Edges connect states that an agent cannot distinguish. In the actual state $w$ of $M$, Cath (agent \(c\)) does not know whether $p$ is true. Ann and Bob, on the contrary, know that $p$ is false. Now suppose that Bob announces that $\neg p$. This truthful public announcement `deletes' all the states where $p$ is true, and the corresponding epistemic indistinguishability relations; in this example, $v$ is `deleted,' and the resulting model is \(M^{\neg p}\). After this announcement Cath knows \(\neg p\), or, formally, \([\neg p] K_c \neg p\). In this paper, within group and coalition announcements, we only quantify over announcements of formulas of the type \(K_a \varphi\). If a group consists only of Cath, who does not know $\neg p$ and hence cannot announce $K_c \neg p$, the following holds in state $w$ of \(M\): \([c, \top] (\neg K_c \neg p \wedge \neg K_c p)\), i.e. whatever \(c\) announces, she still does not know whether \(p\) after the announcement \footnote{For readability, we use $[c]$ rather than $[\{c\}]$ for singleton coalitions.}.  Also, Ann and Bob can remain silent (or announce a tautology \(\top\)) and preclude Cath from knowing that $\neg p$. In other words, there is announcement by their group such that after it is made, agent \(c\) does not know the value of \(p\): \(\langle \{a, b\}, \top \rangle (\neg K_c \neg p \wedge \neg K_c p)\). Moreover, this holds whatever Cath announces at the same time: \(\langle \! [ \{a, b\} ] \! \rangle (\neg K_c \neg p \wedge \neg K_c p)\). On the other hand, a coalition consisting of Ann and Cath does not have such a power, since Bob can always announce that $\neg p$:
\(\neg \langle \! [ \{a, c\} ] \! \rangle (\neg K_c \neg p \wedge \neg K_c p)\), or, equally, \([ \! \langle \{a, c\} \rangle \! ] (K_c \neg p \vee K_c p)\).

Now, we provide formal definitions.

\begin{definition} (Epistemic model)
An \emph{epistemic model} is a triple \(M = (W, \sim, V)\), where
\begin{itemize}
    \item \(W\) is a non-empty set of states;
    \item \(\sim:A \rightarrow \mathcal{P}(W \times W)\) assigns an equivalence relation to each agent; we will denote relation assigned to agent $a \in A$ by $\sim_a$;
    \item \(V:P \rightarrow \mathcal{P}(W)\) assigns a set of states to each propositional variable. 
\end{itemize}
A pair \((W,\sim)\) is called an \emph{epistemic frame}, and a pair \((M,w)\) with \(w \in W\) is called a \emph{pointed model}. An announcement in a pointed model \((M,w)\) results in an \textit{updated pointed model}  
\((M^\varphi, w)\).
Here \(M^\varphi = (W^\varphi, \sim^\varphi, V^\varphi)\), and \(W^\varphi = \llbracket \varphi \rrbracket_M\), \(\sim^\varphi_a = \sim_a \cap\) \((\llbracket \varphi \rrbracket_M \times \llbracket \varphi \rrbracket_M)\), and \(V^\varphi (p) = V(p) \cap \llbracket \varphi \rrbracket_M\). Generally speaking, an updated pointed model \((M^\varphi,w)\) is a restriction of the original one to the states where \(\varphi\) holds. 
\end{definition}

Let \(\mathcal{L}_{EL}^G\) denote the set of formulas of the type \(\bigwedge_{i \in G} K_i \varphi_i\), where for every \(i \in G\) it holds that \(\varphi_i \in \mathcal{L}_{EL}\). We denote elements of  \(\mathcal{L}_{EL}^G\) as $\psi_G$. These are the formulas we will be quantifying
over in modalities of the form $[G, \chi]$ and $[ \! \langle G \rangle \! ]$. 
%In other words, formulae of \(\mathcal{L}_{EL}^{G}\) are of the type `agent \(i\) from coalition \(G\) knows \(\varphi_i\).'

\begin{definition} (Semantics)
\label{def:corgalsemantics}
Let a pointed model \((M,w)\) with \(M = (W\), \(\sim, V)\), \(a \in A\), and \(\varphi\), \(\psi \in \mathcal{L}_{CoRGAL}\) be given.
\begin{center}
\(
\begin{array} {lcl}
(M,w) \models p  &\textrm{iff} &w \in V(p)\\

(M,w) \models \neg \varphi &\textrm{iff} &(M,w) \not \models \varphi\\

(M,w) \models \varphi \wedge \psi &\textrm{iff} &(M,w) \models \varphi \textrm{ and } (M,w) \models \psi\\

(M,w) \models K_a\varphi &\textrm{iff} &\forall v \in W: w \sim_a v \textrm{ implies } (M,v) \models \varphi\\

(M,w) \models [\varphi]\psi &\textrm{iff} &(M,w) \models \varphi \textrm{ implies } (M^\varphi,w) \models \psi\\ 

(M,w) \models [G, \chi] \varphi &\textrm{iff} &(M,w) \models \chi \textrm { and } \forall \psi_G : (M,w) \models  [ \psi_G \wedge \chi ] \varphi\\

(M,w) \models [ \! \langle G \rangle \! ]\varphi &\textrm{iff} &\forall \psi_G \exists \chi_{A \setminus G}:  (M,w) \models \psi_G \rightarrow \langle \psi_G \wedge \chi_{A \setminus G} \rangle \varphi\\

\end{array}\)
\end{center}

\end{definition}
\noindent
 Formula \(\varphi\) is called \textit{valid} if for any pointed model \((M,w)\) it holds that \((M,w) \models \varphi\).  

The semantics for the `diamond' versions of knowledge and public announcement operators (\(\widehat {K}_a \varphi\) and \(\langle \varphi \rangle \psi\)) respectively) are obtained by changing \(\forall\) to \(\exists\) and `implies' to `and' in the corresponding lines. The semantics for duals of relativised group announcements and coalition announcements is as follows:

\begin{center}
\(
\begin{array}{lcl}
(M,w) \models \langle G, \chi \rangle \varphi &\textrm{iff} &(M,w) \models \chi \textrm { implies } \exists \psi_G: (M,w) \models  \langle \psi_G \wedge \chi \rangle \varphi\\

(M,w) \models \langle \! [ G ] \! \rangle \varphi &\textrm{iff} &\exists \psi_G \forall \chi_{A\setminus G}: (M,w) \models \psi_G \wedge [ \psi_G \wedge \chi_{A \setminus G} ] \varphi,\\

\end{array}\)
\end{center}
The existential version of the coalition announcement operator is read as `there is an announcement by agents from \(G\), such that whatever other agents \(A \setminus G\) announce at the same time, \(\varphi\) holds.'

Note that semantics of coalition announcement operators are given in a `classic' way. An equivalent definition is possible using relativised group announcements.

\begin{center}
\(
\begin{array} {lcl}

(M,w) \models [ \! \langle G \rangle \! ] \varphi &\textrm{iff} &\forall \psi_G: (M,w) \models \langle A \setminus G, \psi_G \rangle \varphi \\  

(M,w) \models \langle \! [ G ] \! \rangle \varphi &\textrm{iff} &\exists \psi_G: (M,w) \models [A \setminus G, \psi_G] \varphi

\end{array}\)
\end{center}

We can use relativised group announcements to define classic group announcements: $[G] \varphi \leftrightarrow [G, \top] \varphi$ and $\langle G \rangle \varphi \leftrightarrow \langle G, \top \rangle \varphi$. 

Following \cite{balbiani07, balbiani08, agotnes10, agotnes08, agotnes16, balbiani15, vanditmarsch12, agotnes14} we restrict formulas which agents in a group or coalition can announce to formulas of \( \mathcal{L}_{EL}\).
This allows us to avoid circularity in the definition.

\subsection{Axiomatisation and Some Logical Properties}
In this section we present an axiomatisation of \(\mathbf{CoRGAL}\) and show its soundness. It is based on the axiom systems for \textbf{PAL}, and have two additional axioms and four additional rules of inference. 

\begin{definition}
\label{def::axiomatisation}
The \emph{axiom system for} \textbf{CoRGAL} is an extension of \textbf{PAL} with a relativised version of \textbf{GAL} and interaction axioms.
\begin{center}
$  \begin{array}[t]{ll}
    (A0) &\textrm{propositional tautologies},  \\

    (A1) &K_a(\varphi \rightarrow \psi) \rightarrow (K_a \varphi \rightarrow K_a \psi), \\      

    (A2) &K_a\varphi \rightarrow \varphi, \\

        (A3) &K_a\varphi \rightarrow K_a K_a \varphi,  \\  

        (A4) &\neg K_a \varphi \rightarrow K_a \neg K_a \varphi,  \\

        (A5) &[\varphi]p \leftrightarrow (\varphi \rightarrow p),  \\

        (A6) &[\varphi] \neg \psi \leftrightarrow (\varphi \rightarrow \neg [\varphi] \psi),  \\

        (A7) &[\varphi](\psi \wedge \chi) \leftrightarrow ([\varphi]\psi \wedge [\varphi]\chi),\\

        (A8) &[\varphi]K_a\psi \leftrightarrow (\varphi \rightarrow K_a[\varphi]\psi), \\

        (A9) &[\varphi][\psi]\chi \leftrightarrow [\varphi \wedge [\varphi]\psi] \chi, \\       
    
  	(A10) &[G, \chi]\varphi \rightarrow \chi \wedge [\psi_G \wedge \chi]\varphi \textrm{ for any } \psi_G,\\

 	(A11) &[ \! \langle G \rangle \! ] \varphi \rightarrow \langle A \setminus G, \psi_G \rangle \varphi \textrm{ for any } \psi_G, \\

 	(R0) &\textrm{If } \vdash \varphi \textrm{ and } \vdash \varphi \rightarrow \psi, \textrm{ then } \vdash \psi, \\

	(R1) &\textrm{If }\vdash \varphi, \textrm{ then } \vdash K_a \varphi, \\

	(R2) &\textrm{If }\vdash \varphi, \textrm{ then } \vdash [\psi] \varphi, \\

	(R3) &\textrm{If }\vdash \varphi, \textrm{ then } \vdash [G, \chi] \varphi, \\

	(R4) &\textrm{If }\vdash \varphi, \textrm{ then } \vdash [ \! \langle G \rangle \! ] \varphi, \\

	(R5) &\textrm{If } \forall \psi_G: \vdash \eta (\chi \wedge [\psi_G \wedge \chi] \varphi), \textrm{ then } \vdash \eta ([G, \chi ]\varphi), \\

 	(R6) &\textrm{If } \forall \psi_G:  \vdash \eta (\langle A \setminus G, \psi_G \rangle \varphi), \textrm{ then } \vdash \eta ([ \! \langle G \rangle \! ]\varphi).\\

\end{array}$
\end{center}

\end{definition}

So, \(\mathbf{CoRGAL}\) is the smallest subset of \(\mathcal{L}_{CoRGAL}\) that contains all the axioms \(A0\) -- \(A11\) and is closed under rules of inference \(R0\) -- \(R6\). Elements of \(\mathbf{CoRGAL}\) are called \textit{theorems}. Note that \(R5\) and \(R6\) are infinitary rules: they require an infinite number of premises. Finding finite axiomatisations of any of \(\mathbf{APAL}\), \(\mathbf{GAL}\), or \(\mathbf{CAL}\) is an open problem. Note also that \(\mathbf{CoRGAL}\) includes coalition logic \cite{pauly02}, that is all the axioms of the latter are validities of \(\mathbf{CoRGAL}\) and a rule of inference preserves validity (see Appendix A).

\begin{definition} (Soundness and completeness)
	An axiomatisation is \textit{sound}, if for any formula \(\varphi\) of the language, it holds that \(\varphi \in \mathbf{CoRGAL}\) implies \(\varphi\) is valid. And vice versa for \textit{completeness}.
\end{definition}

Soundness of \(A0\)--\(A4\), \(R0\), and \(R1\) is due to soundness of \(\mathbf{S5}\). Axioms \(A5\)--\(A9\) and rule of inference \(R3\) are sound, since \(\mathbf{PAL}\) is sound \cite{del}. We show soundness of $R3$--$R6$ in Proposition \ref{prop::soundness}, and validity of \(A10\) and $A11$ in Proposition \ref{prop:validity}.

\begin{proposition}
\label{prop:validity}
Axioms $A10$ and $A11$ are valid.
\end{proposition}

\begin{proof}
Follows directly from the definition of semantics (Definition \ref{def:corgalsemantics}). We just show validity of $(A11)$.

Assume that for some arbitrary pointed model $(M,w)$ it holds that $(M,w) \models [ \! \langle G \rangle \! ] \varphi$. By semantics this is equivalent to $\forall \psi_G, \exists \chi_{A \setminus G}$: $(M,w) \models \psi_G \rightarrow \langle \psi_G \wedge \chi_{A \setminus G} \rangle \varphi$. Since $\psi_G$ quantifies over all epistemic formulas known to $G$, we can choose any particular $\psi_G$. Hence, we have that $\exists \chi_{A \setminus G}$: $(M,w) \models \psi_G \rightarrow \langle \psi_G \wedge \chi_{A \setminus G} \rangle \varphi$, which is equivalent to $(M,w) \models \langle A \setminus G, \psi_G \rangle \varphi$ by semantics.
\end{proof}

\begin{proposition}
\label{prop::soundness}
    $R3$, $R4$, $R5$, and $R6$ are sound, that is, they preserve validity.
\end{proposition}
\begin{proof}
A proof is given in Appendix B (Proposition \ref{prop::app1}).
 \end{proof}
 
 Next, we show some properties of \textbf{CAL} and \textbf{GAL}.
 
Whether \textbf{CAL} operators can be expressed in \textbf{GAL} is an open question. The most probable definition of coalition announcements in terms of group announcements is $\langle \! [ G ] \! \rangle \varphi \leftrightarrow \langle G \rangle [A \setminus G] \varphi$. Validity of this formula was stated to be an open question in \cite{vanditmarsch12,agotnes16}. We partially settle this problem by proving one direction.

Consider the left-to-right direction of the formula. In the antecedent all agents make a simultaneous announcement, whereas in the consequent agents from $A \setminus G$ know the announcement $\psi_G$ made by $G$. Thus, in the updated model $(M^{\psi_G},w)$ the agents in $A \setminus G$ may have learned some \textit{new} epistemic formulas $\chi_{A \setminus G}$ that they did not know before the announcement. However, since $\psi_G$ holds in the initial model, and $\chi_{A \setminus G}$ holds in the updated one, agents from $A \setminus G$ can always make an announcement in the initial model that they know that after announcement of $\psi_G$, $\chi_{A \setminus G}$ is true. 

 Returning to our example (Figure \ref{fig::tark1}), whichever formulae \(\psi_1\) and \(\psi_2\) Ann and Bob announce, and whichever formula \(\varphi\) Cath learns afterwards, she can always announce \([\psi_1 \wedge \psi_2] K_c \varphi\) simultaneously with them in the initial situation. Informally, if after Bob's announcement of $\neg p$, 
Cath learns that $\neg p$, she can announce: `If you say that \(\neg p\) holds, then I will know it,' or $[\neg p] K_c \neg p$. We use this idea to prove that if the agents in $A \setminus G$ can prevent $\varphi$ after the announcement by $G$, then they could have prevented it before. 

Due to restriction of announcements to formulas of epistemic logic, we cannot directly employ public announcement operators in agents' `utterances.' In order to avoid this, we use the standard translation of $\mathbf{PAL}$ into epistemic logic.
\begin{definition}
\label{def:translation}
\emph{Translation function} \(t:\mathcal{L}_{PAL} \rightarrow \mathcal{L}_{EL}\) \cite{del} is defined as follows:
\begin{center}
\(
\begin{array}[t] {lcl}
t(p) &= &p, \\

t(\neg \varphi) &= &\neg t(\varphi), \\

t(\varphi \wedge \psi) &= &t(\varphi) \wedge t(\psi), \\

t(K_a \varphi) &= &K_at(\varphi), \\

t([\varphi]p) &= &t(\varphi \rightarrow p),\\

\end{array}\)
\( \begin{array}[t]{lcl}

t([\varphi] \neg \psi) &= &t(\varphi \rightarrow \neg [\varphi]\psi),\\

t([\varphi](\psi \wedge \chi)) &= &t([\varphi] \psi \wedge [\varphi] \chi),\\

t([\varphi]K_a \psi) &= &t(\varphi \rightarrow K_a [\varphi]\psi),\\

t([\varphi][\psi]\chi) &= &t([\varphi \wedge [\varphi]\psi]\chi).\\
\end {array} \)
\end{center}

Every \(\varphi \in \mathcal{L}_{PAL}\) is equivalent to \(t(\varphi) \in \mathcal{L}_{EL}\).
\end{definition}

\begin{proposition}
\label{prop:axiom}
    $\langle \! [ G ] \! \rangle  \varphi \rightarrow \langle G \rangle [ A \setminus G ] \varphi$ is valid.
\end{proposition}

\begin{proof}

Assume that for some pointed model $(M,w)$ it holds that $(M,w) \models \langle \! [ G ] \! \rangle \varphi$. By the semantics of \textbf{CAL} this is equivalent to \[\exists \psi_G, \forall \chi_{A \setminus G}: (M,w) \models \psi_G \wedge [\psi_G \wedge \chi_{A \setminus G}] \varphi.\] Since $\chi_{A \setminus G}$ quantifies over all possible announcements by $A \setminus G$, it also quantifies over a specific subset of these announcements --- $K_{A \setminus G} [\psi_G] \chi^\prime_{A \setminus G} := \bigwedge_{a \in A \setminus G} K_a [\psi_G] \chi^\prime_a$ for some $\psi_G$ and for all $\chi^\prime_{a} \in \mathcal{L}_{EL}$.

Hence $\exists \psi_G, \forall \chi_{A \setminus G}$: $(M,w) \models \psi_G \wedge [\psi_G \wedge \chi_{A \setminus G}] \varphi$ implies
\[\exists \psi_G, \forall \chi^\prime_{A \setminus G}: (M,w) \models \psi_G \wedge [\psi_G \wedge K_{A \setminus G} [\psi_G] \chi^\prime_{A \setminus G}] \varphi.\]

Note that $K_{A \setminus G} [\psi_G] \chi^\prime_{A \setminus G}$ is not an epistemic formula \emph{per se}. It is equivalent, however, to an epistemic formula of type $K_{A \setminus G} \chi_{A \setminus G}$, where $\chi_{A \setminus G} \in \mathcal{L}_{EL}$, via translation $t(K_{A \setminus G} [\psi_G] \chi^\prime_{A \setminus G})$ (Definition \ref{def:translation}). Thus we have that 
\[\exists \psi_G, \forall \chi^\prime_{A \setminus G}: (M,w) \models \psi_G \wedge [\psi_G \wedge t(K_{A \setminus G} [\psi_G] \chi^\prime_{A \setminus G})] \varphi.\]

Let us consider announcement $\psi_G \wedge t(K_{A \setminus G} [\psi_G] \chi^\prime_{A \setminus G})$. By propositional reasoning it is equivalent to $\psi_G \wedge (\psi_G \rightarrow t(K_{A \setminus G} [\psi_G] \chi^\prime_{A \setminus G}))$. Since $\psi_G$ is an epistemic formula, the latter is equivalent to $\psi_G \wedge t(\psi_G \rightarrow K_{A \setminus G} [\psi_G] \chi^\prime_{A \setminus G})$. Applying the \textbf{PAL} axiom $[\psi] K_a \varphi \leftrightarrow (\psi \rightarrow K_a [\psi] \varphi)$, we get $\psi_G \wedge t([\psi_G] K_{A \setminus G} \chi^\prime_{A \setminus G})$, which is equivalent to $\psi_G \wedge [\psi_G] K_{A \setminus G} \chi^\prime_{A \setminus G}$. Finally, we have that
\[\exists \psi_G, \forall \chi^\prime_{A \setminus G}: (M,w) \models \psi_G \wedge [\psi_G \wedge [\psi_G] K_{A \setminus G} \chi^\prime_{A \setminus G}] \varphi.\]

Using the axiom $[\psi][\chi]\varphi \leftrightarrow [\psi \wedge [\psi] \chi] \varphi$, we get 
\[\exists \psi_G, \forall \chi^\prime_{A \setminus G}: (M,w) \models \psi_G \wedge [\psi_G] [K_{A \setminus G} \chi^\prime_{A \setminus G}] \varphi,\]
where $\chi^\prime_{A \setminus G} \in \mathcal{L}_{EL}$.
The latter is equivalent $(M,w) \models \langle G \rangle [A \setminus G] \varphi$ due to validity $\models \psi \wedge [\psi] \varphi \leftrightarrow \langle \psi \rangle \varphi$ and by the semantics of \textbf{GAL}.
\end{proof}

Intuition suggests that various groups and coalitions of agents, when united, can do no worse than if they were acting on their own. In the remaining part of this section we show that this intuition is indeed true.

We start with a somewhat obvious statement: if some configuration of a model can be achieved by a coalition, then the configuration can be achieved by a superset of the coalition.

\begin{proposition}
\label{prop:one}
$\langle \! [ G ] \! \rangle \varphi \rightarrow \langle \! [ G \cup H ] \! \rangle \varphi $, where $G,H \subseteq A$, is valid.
\end{proposition}

\begin{proof}
Appendix, Proposition \ref{prop:oneapp}.
\end{proof}

It was shown in \cite{agotnes10} that $\langle G \rangle \varphi \leftrightarrow \langle G \rangle \langle G \rangle \varphi$. This property demonstrates that within the framework of \textbf{GAL} a multiple-step strategy of a group can be executed in a single step.  Whether this is true for \textbf{CAL} is an open question. We show, however, that if truth of some $\varphi$ can be achieved by two consecutive coalition announcements by $G$, then whatever agents from $A \setminus G$ announce, they cannot preclude $G$ from making $\varphi$ true.

\begin{proposition}
\label{prop:two}
$\langle \! [ G ] \! \rangle \langle \! [ G ] \! \rangle \varphi \rightarrow [ \! \langle A \setminus G \rangle \! ] \varphi$ is valid.
\end{proposition}

\begin{proof}
Appendix, Proposition \ref{prop:twoapp}.
\end{proof}

Whether $\langle \! [ G ] \! \rangle \langle \! [ G ] \! \rangle \varphi \rightarrow \langle \! [ G ] \! \rangle \varphi$ is valid is an open question. We conjecture that the property is not valid. Consider $\langle \! [ G ] \! \rangle \langle \! [ G ] \! \rangle \varphi$: after initial announcement, coalition $G$ has a consecutive announcement to make $\varphi$ true. This announcement, however, depends on the choice of $A \setminus G$ in the first operator. In other words, consecutive announcement by $G$ may vary depending on the initial announcement by $A \setminus G$. Hence, it seems highly counterintuitive that $G$ has a single announcement that can incorporate all possible simultaneous announcements by $A \setminus G$ in a general (infinite) case.  

Formula $\langle G \rangle \langle H \rangle \varphi \rightarrow \langle G \cup H \rangle \varphi$ is a validity of \textbf{GAL} \cite{agotnes10}. Again, it is unknown whether the same property holds for coalition operators, and, for the same reasons as for Proposition \ref{prop:two}, we conjecture that the corresponding formula is not valid in \textbf{CAL}.

\begin{proposition}
\label{prop::app2}
    $\langle \! [ G ] \! \rangle \langle \! [ H ] \! \rangle \varphi  \rightarrow [ \! \langle A \setminus (G \cup H) \rangle \! ] \varphi\) is valid.
\end{proposition}
\begin{proof}
Let $(M,w) \models \langle \! [ G ] \! \rangle \langle \! [ H ] \! \rangle \varphi$. By Proposition \ref{prop:one} applied twice, we have $(M,w) \models \langle \! [ G \cup H ] \! \rangle \langle \! [ G \cup H ] \! \rangle \varphi$, and by Proposition \ref{prop:two}, the latter implies $(M,w) \models [ \! \langle A \setminus (G \cup H) \rangle \! ] \varphi$.
\end{proof}

\iffalse
\begin{lemma}
\label{lemma:help}
If $\varphi_1 \rightarrow \varphi_2$ is valid, then $\langle \! [ G ] \! \rangle \varphi_1 \rightarrow \langle \! [ G ] \! \rangle \varphi_2$ is valid as well.
\end{lemma}

\begin{proof}
Assume that $\varphi_1 \rightarrow \varphi_2$ is a validity of CAL. Also, let $(M,w) \models \langle \! [ G ] \! \rangle \varphi_1$ for some arbitrary $(M,w)$. By the semantics of CAL, we have that $\exists \psi_G, \forall \chi_{A \setminus G}:$ $(M,w) \models \psi_G \wedge [\psi_G \wedge \chi_{A \setminus G}] \varphi_1$. The latter is equivalent to $\exists \psi_G, \forall \chi_{A \setminus G}:$ $(M,w) \models \psi_G$ and $(M,w) \models \psi_G \wedge \chi_{A \setminus G}$ implies $(M,w)^{\psi_G \wedge \chi_{A \setminus G}} \models \varphi_1$. Since $\varphi_1 \rightarrow \varphi_2$ and by modus ponens, we have $\exists \psi_G, \forall \chi_{A \setminus G}:$ $(M,w) \models \psi_G$ and $(M,w) \models \psi_G \wedge \chi_{A \setminus G}$ implies $(M,w)^{\psi_G \wedge \chi_{A \setminus G}} \models \varphi_2$, which is $(M,w) \models \langle \! [ G ] \! \rangle \varphi_2$ by the semantics.
\end{proof}
\fi

Next, we show that splitting an announcement by a unified coalition into consecutive announcements of sub-coalitions may decrease their power to force certain outcomes.  Whether  \(\langle \! [ G \cup H ] \! \rangle \varphi \rightarrow \langle \! [ G ] \! \rangle \langle \! [ H ] \! \rangle \varphi\) is valid was mentioned as an open question in \cite{agotnes16}. We settle this problem by presenting a counterexample.

\begin{proposition}
\label{prop:calcounter}
    \(\langle \! [ G \cup H ] \! \rangle \varphi \rightarrow \langle \! [ G ] \! \rangle \langle \! [ H ] \! \rangle \varphi\) is not valid.
\end{proposition}

\begin{proof}
    Let $G = \{a\}, H = \{b\}$, and $\varphi:= K_b(p \wedge q \wedge r) \wedge \neg K_a(p \wedge q \wedge r) \wedge \neg K_c(p \wedge q \wedge r)$. Formula $\varphi$ says that agent $b$ knows that the given propositional variables are true, and agents $a$ and $c$ do not. Consider model $(M,pqr)$ in Figure \ref{fig:counter} (reflexive and transitive arrows are omitted for convenience). Names of the states in the model show values of propositional variables; for example, $(M,p \overline{q} r) \models p \wedge \neg q \wedge r$.
    
    \begin{figure} [h]
\centering
\begin{tikzpicture}

\node (3) at (0,0) {\fbox{$pqr$}};
\node (4) at (0,-2) {$pq \overline{r}$};
\node (5) at (-2,-2) {$\overline{p} qr$};
\node (6) at (2,-2) {$p \overline{q} r$};

\draw (3) -- node[right] {$c$} (4);
\draw (3) -- node[left] {$a,c$} (5);
\draw (4) -- node[below] {$c$} (5); 
\draw (3) -- node[right] {$b$} (6);
\end{tikzpicture}
\caption{Counterexample}
\label{fig:counter}
\end{figure}
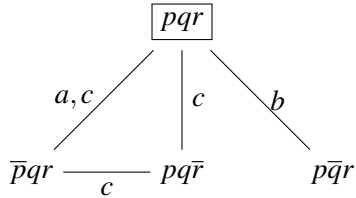
    
    By the semantics $(M, pqr) \models \langle \! [ \{ a, b \} ] \! \rangle \varphi$ if and only if $\exists \psi_a, \exists \psi_b, \forall \chi_c: (M, pqr) \models \psi_a \wedge \psi_b \wedge [\psi_a \wedge \psi_b \wedge \chi_c] \varphi$. Let $\psi_a$ be $K_a q$, and $\psi_b$ be $K_b \top$. 
    Observe that $(M, pqr) \models K_a q \wedge K_b \top$. This announcement leads to $b$ learning that $q$.
    Moreover, \(c\) does not know any formula that she can announce to avoid \(\varphi\). An informal argument is as follows. 
    By announcing $K_a q$ agent $a$ `chooses' a union of $a$-equivalence classes $\{pqr, \overline{p}qr, pq\overline{r}\}$ (and $b$ `chooses' the whole model). Any simultaneous `choice' of $c$ includes $\{pqr, \overline{p}qr, pq\overline{r}\}$ as a subset. Thus, intersection of $\{pqr, \overline{p}qr, pq\overline{r}\}$ and any of unions of $c$-equivalence classes is $\{pqr, \overline{p}qr, pq\overline{r}\}$, and $\varphi$ is true in such a restriction of the model.
    
   Let us show that $(M, pqr) \not \models \langle \! [ \{ a \} ] \! \rangle \langle \! [ \{ b \} ] \! \rangle \varphi$, or, equivalently, $(M, pqr) \models [ \! \langle \{ a \} \rangle \! ] [ \! \langle \{ b \} \rangle \! ] \neg \varphi$. According to the semantics, $\forall \psi_a, \exists \chi_b, \exists \chi_c$: $(M, pqr) \models \psi_{a} \rightarrow \langle \psi_a \wedge \chi_b \wedge \chi_c \rangle [ \! \langle \{ b \} \rangle \! ] \neg \varphi$.
  Assume that for an arbitrary $\psi_a$, announcements by $b$ and $c$ are $K_b p$ and $K_c \top$ correspondingly. Then $(M, pqr) \models \psi_a \wedge [ \psi_a \wedge K_b p \wedge K_c \top ] [ \! \langle \{ b \} \rangle \! ] \neg \varphi$.
    Note that no matter what $a$ announces, $K_b p$ `forces' her to learn that $p \wedge q \wedge r$, and whatever is announced in the updated model $(M^{\psi_a \wedge K_b p \wedge K_c \top }, pqr)$, $a$'s knowledge of $p \wedge q \wedge r$ and, hence, falsity of $\varphi$ remains. Thus we reached a contradiction.
\end{proof}

The same counterexample can be used to demonstrate that $[ \! \langle A \setminus (G \cup H) \rangle \! ] \varphi \rightarrow \langle \! [ G ] \! \rangle \langle \! [ H ] \! \rangle \varphi$ is not valid, where $A \setminus (G \cup H) = \{c\}$. In the proof the Proposition \ref{prop:calcounter}
we show that $(M, pqr) \models \langle \! [ \{ a, b \} ] \! \rangle \varphi$. Using validity $\langle \! [ G ] \! \rangle \varphi \rightarrow [ \! \langle A \setminus G \rangle \! ] \varphi$ we obtain $(M, pqr) \models [ \! \langle c \rangle \! ] \varphi$. The rest of the proof remains the same.

\begin{corollary}
$[ \! \langle A \setminus (G \cup H) \rangle \! ] \varphi \rightarrow \langle \! [ G ] \! \rangle \langle \! [ H ] \! \rangle \varphi$ is not valid.
\end{corollary}

\subsection{Completeness}

In order to prove completeness of \(\mathbf{CoRGAL}\), we expand and modify the completeness proof for \(\mathbf{APAL}\) \cite{balbiani08, balbiani15, balbiani15a}. Although the proof is partially based upon the classic canonical model approach, we have to ensure that construction of maximal consistent theories (Proposition \ref{prop:lindenbaum}) allows us to include infinite amount of formulas for cases of coalition announcements. This is possible due to axioms $A10$, $A11$ and rules of inference $R5,R6$. After that we use induction on complexity of \textbf{CoRGAL} formulas to prove the Truth Lemma. 

First, we prove a useful auxiliary lemma.

\begin{lemma}
\label{lemma::2}
	Let \(\varphi, \psi \in \mathcal{L}_{CoRGAL}\). If \(\varphi \rightarrow \psi\) is a theorem, then \(\eta (\varphi) \rightarrow \eta (\psi)\) is a theorem as well.
\end{lemma}

\begin{proof}
    Appendix, Lemma \ref{lemma::2app}.
\end{proof}

Now, the first part of the proof up to Proposition \ref{prop:lindenbaum} is based on \cite{balbiani08}. Here we introduce theories and prove the Lindenbaum Lemma.

\begin{definition}
\label{def:theory}
    A set of formulas \(x\) is called a \emph{theory} if and only if it contains \(\mathbf{CoRGAL}\), and is closed under \(R0, R5,\) and \(R6\). A theory \(x\) is consistent if and only if \(\bot \not \in x\), and is maximal if and only if for all \(\varphi \in \mathcal{L}_{CoRGAL}\) it holds that either \(\varphi \in x\) or \(\neg \varphi \in x\). 
\end{definition}

Note that theories are not closed under necessitation rules. The reason for this is that while these rules preserve validity, they do not preserve truth, whereas $R0, R5$, and $R6$ preserve both validity and truth.

\begin{proposition}
\label{prop::mct}
    Let \(x\) be a theory, \(\varphi, \psi \in \mathcal{L}_{CoRGAL}\), and \(a \in A\). The following are theories: \(x + \varphi = \{\psi: \varphi \rightarrow \psi \in x\}, K_a x = \{\varphi: K_a \varphi \in x\}\), and \([\varphi]x = \{\psi: [\varphi]\psi \in x\}\).
\end{proposition}

\begin{proof}
Appendix, Proposition \ref{prop::mctapp}.
\end{proof}

\begin{proposition}
\label{prop::consistency}
Let \(\varphi \in \mathcal{L}_{CoRGAL}\). Then \(\mathbf{CoRGAL} + \varphi\) is consistent iff \(\neg \varphi \not \in \mathbf{CoRGAL}\).
\end{proposition}

\begin{proof}
Appendix, Proposition \ref{prop::consistencyapp}
\end{proof} 

The following proposition is a variation of the Lindenbaum Lemma. In order to prove it, we rely heavily on rules of inference $R5$ and $R6$.

\begin{lemma}[Lindenbaum]
\label{prop:lindenbaum}
    Every consistent theory \(x\) can be extended to a maximal consistent theory \(y\).
\end{lemma}

\begin{proof}
    Let \(\psi_0, \psi_1, \mathellipsis\) be an enumeration of formulas of the language, and let \(y_0 = x\). 
    Suppose that for some \(n \geq 0\), \(y_n\) is a consistent theory, and \(x \subseteq y_n\). 
    If \(y_n + \psi_n\) is consistent, then \(y_{n+1} = y_n + \psi_n\). 
    Otherwise, if \(\psi_n\) is not a conclusion of either $R5$ or $R6$, \(y_{n+1} = y\). 

    If \(\psi_n\) is a conclusion of $R5$, we enumerate all the subformulas of \(\psi_n\) which contain relativised group announcement modalities \([G, \chi]\). 
    Let $\eta_1 ([G, \chi]\varphi_1),$ $\mathellipsis,$ $\eta_k ([G, \chi]\varphi_k)$ be all these subformulas. 
    Then \(y_n^0, \mathellipsis, y_n^k\) is a sequence of consistent theories, where \(y_n^0 = y_n\), and for some \(i < k\), \(y_n^i\) is a consistent theory containing \(y_n\) and \(\neg \eta_i ([G, \chi]\varphi_i)\). 
    Since \(y_n^i\) is closed under $R5$, there exists $\psi_G$ such that \(\eta_i (\chi \wedge [\psi_G \wedge \chi]\varphi_i) \not \in y_n^i\). 
    Hence, \(y_n^{i+1} = y_n^i + \neg \eta_i (\chi \wedge [\psi_G \wedge \chi]\varphi_i)\), and \(y_{n+1} = y_n^k\). Note that adding such a witness $\psi_G$ corresponds to the semantics of relativised group announcements, i.e. for formula $\eta_i \{\langle G, \chi \rangle \neg \varphi_i \}$ we have $\psi_G$ such that $\eta_i \{\chi \rightarrow \langle \psi_G \wedge \chi \rangle \neg \varphi_i\}$. 
    
    Now we consider the case when \(\psi_n\) is a conclusion of $R6$. We enumerate all the subformulas of \(\psi_n\) which contain coalition announcement modalities \([ \! \langle G \rangle \! ]\). 
    Let \(\eta_1 ([ \! \langle G \rangle \! ]\varphi_1)\), \(\mathellipsis\), \(\eta_k ([ \! \langle G \rangle \! ]\varphi_k)\) be all these subformulas. 
    Then \(y_n^0, \mathellipsis, y_n^k\) is a sequence of consistent theories, where \(y_n^0 = y_n\), and for some \(i < k\), \(y_n^i\) is a consistent theory containing \(y_n\) and \(\neg \eta_i ([ \! \langle G \rangle \! ]\varphi_i)\).
    Since \(y_n^i\) is closed under $R6$, there exists $\psi_{G}$ such that \(\eta_i (\langle A \setminus G, \psi_G \rangle \varphi_i)\) \(\not \in y_n^i\).   
    Hence, \(y_n^{i+1} = y_n^i + \neg \eta_i (\langle A \setminus G, \psi_G \rangle  \varphi_i)\), and \(y_{n+1} = y_n^k\).
    Note that since for all $\chi_{A \setminus G}$: \(\eta ([A \setminus G, \psi_G] \varphi) \rightarrow \eta (\psi_G \wedge [\psi_G \wedge \chi_{A \setminus G}] \varphi)\) are theorems, these formulas and their contrapositions (due to Proposition \ref{lemma::2}) are already in \(y_n^i\) (because \(y_n^i\) is a theory). Thus, adding \(\neg \eta_i ( \langle A \setminus G, \psi_G \rangle \varphi_i)\) to \(y^i_n\) adds all the \(\neg \eta_i (\psi_G \rightarrow \langle \psi_G \wedge \chi_{A \setminus G} \rangle \varphi_i)\) for $\chi_{A \setminus G}$ as well. This satisfies the semantics of coalition announcements, i.e. for formula $\eta_i \{\langle \! [ G ] \! \rangle \neg \varphi_i\}$ we have some $\psi_G$ such that for all $\chi_{A \setminus G}$: $\eta_i \{\psi_G \wedge [\psi_G \wedge \chi_{A \setminus G}] \neg \varphi_i\}$.
    
    Finally, \(y\) is a maximal consistent theory, and \(x \subseteq y\).
\end{proof} 

The rest of the proof is an expansion of the one from \cite{balbiani15}. It employs induction on complexity of formulae to prove the Truth Lemma (Proposition \ref{prop::truth}) and, ultimately, completeness (Proposition \ref{prop::completeness}) of \(\mathbf{CoRGAL}\).

\begin{definition}
    The \emph{size} of some formula \(\varphi \in \mathcal{L}_{CoRGAL}\) is defined as follows:
    \begin{enumerate}
        \item \(Size (p) = 1\),
        \item \(Size (\neg \varphi)\) \(=\) \(Size (K_a \varphi) = Size ([G, \chi] \varphi) =\) \(Size ([ \! \langle G \rangle \! ] \varphi) =\) \(Size (\varphi) + 1\), 
        \item \(Size (\varphi \wedge \psi) = Size (\varphi) + Size (\psi) + 1\),
        \item \(Size ([\psi] \varphi) = Size (\psi) + 3 \cdot Size (\varphi)\).
    \end{enumerate}
    
    The \([,]\)-\emph{depth} is defined as follows:
    \begin{enumerate}
        \item \(d_{[,]} (p) = 0\),
        \item \(d_{[,]} (\neg \varphi) = d_{[,]} (K_a \varphi) = d_{[,]} ([ \! \langle G \rangle \! ] \varphi) = d_{[,]} (\varphi)\),
        \item \(d_{[,]} (\varphi \wedge \psi) = \mathrm{max} \{d_{[,]} (\varphi), d_{[,]} (\psi)\}\),
        \item \(d_{[,]} ([\psi] \varphi) = d_{[,]} (\psi) + d_{[,]} (\varphi)\),
        \item \(d_{[,]} ([G, \chi]\varphi) = d_{[,]} (\chi) + d_{[,]} (\varphi) + 1\).
    \end{enumerate}
    
    The \([ \! \langle \! \rangle \! ]\)-\emph{depth} is the same as \([,]\), with the following exceptions.  
    \begin{enumerate}
    \item \(d_{[ \! \langle \! \rangle \! ]} ([G, \chi]\varphi) = d_{[ \! \langle \! \rangle \! ]} (\chi) + d_{[ \! \langle \! \rangle \! ]} (\varphi)\),
    \item \(d_{[ \! \langle \! \rangle \! ]} ([ \! \langle G \rangle \! ]\varphi) = d_{[ \! \langle \! \rangle \! ]} (\varphi) + 1\).
    \end{enumerate}
\end{definition}

\begin{definition}
	The binary relation \(<^{Size}_{[,], [ \! \langle \! \rangle \! ]}\) between \(\varphi, \psi \in \mathcal{L}_{CoRGAL}\) is defined as follows: \(\varphi <^{Size}_{[,], [ \! \langle \! \rangle \! ]} \psi\) iff \(d_{[ \! \langle \! \rangle \! ]} (\varphi) < d_{[ \! \langle \! \rangle \! ]} (\psi)\), or, otherwise,  \(d_{[ \! \langle \! \rangle \! ]} (\varphi) = d_{[ \! \langle \! \rangle \! ]} (\psi)\), and either \(d_{[,]} (\varphi) < d_{[,]} (\psi)\), or \(d_{[,]} (\varphi) = d_{[,]} (\psi)\) and \(Size (\varphi) < Size (\psi)\). The relation is a well-founded strict partial order between formulae. Note that for all epistemic formulas $\psi$ we have that $d_{[,]} (\psi) = d_{[ \! \langle \! \rangle \! ]} (\psi)  = 0$.
\end{definition}
  
  We need the following proposition the for Truth Lemma.
  \begin{proposition}
  \label{prop:measure}
  Let $\psi_G$, $G \subseteq A$, and $\chi, \varphi, \tau \in \mathbf{CoRGAL}$.
\begin{enumerate}
\item $\chi \wedge [\psi_G \wedge \chi]\varphi <^{Size}_{[,], [ \! \langle \! \rangle \! ]} [G, \chi]\varphi$,

\item $[\tau] (\chi \wedge [\psi_G \wedge \chi]\varphi) <^{Size}_{[,], [ \! \langle \! \rangle \! ]} [\tau][G, \chi]\varphi$,

\item $\langle A \setminus G, \psi_G \rangle \varphi <^{Size}_{[,], [ \! \langle \! \rangle \! ]} [ \! \langle G \rangle \! ] \varphi$,

\item $[\tau] \langle A \setminus G, \psi_G \rangle \varphi <^{Size}_{[,], [ \! \langle \! \rangle \! ]} [\tau] [ \! \langle G \rangle \! ] \varphi$.
\end{enumerate}
\end{proposition}

\begin{proof}
Appendix, Proposition \ref{prop:measureapp}.
\end{proof}

\begin{definition}
    The \emph{canonical model} is the model \(M^C = (W^C, \sim^C, V^C)\), where
    \begin{itemize}
        \item \(W^C\) is the set of all maximal consistent theories,
        \item \(\sim^C\) is defined as \(x \sim^C_a y\) iff \(K_a x = K_a y\),
        \item \(x \in V^C(p)\) iff \(p \in x\).
    \end{itemize}
    Relation \(\sim^C\) is equivalence due to axioms \(A2\), \(A3\), and \(A4\).
\end{definition}

\begin{definition}
    Let \(\varphi \in \mathcal{L}_{CoRGAL}\). 
    Condition \(P(\varphi)\): for all maximal consistent theories \(x\), \(\varphi \in x\) iff $(M^C, x) \models \varphi$. 
    Condition \(H(\varphi)\): for all \(\psi \in\mathcal{L}_{CoRGAL}\), if \(\psi <^{Size}_{[,], [ \! \langle \! \rangle \! ]} \varphi\), then \(P(\psi)\). 
\end{definition}

\begin{proposition}
\label{prop::ph}
    For all \(\varphi \in \mathcal{L}_{CoRGAL}\), if \(H(\varphi)\), then \(P(\varphi)\).
\end{proposition}

\begin{proof}
    Suppose \(H(\varphi)\) holds, and let \(x\) be a maximal consistent theory. The proof is by induction on \(<^{Size}_{[,], [ \! \langle \! \rangle \! ]}\)-complexity of formulae. 
    Most of the cases were proved in \cite{balbiani15}.
    We prove here only the remaining instances involving realtivised group and coalition announcements.
    
    \emph{Case} \(\varphi_0 = [G, \chi] \varphi\). Suppose that \([G, \chi] \varphi \in x\). Since $x$ is a theory, by axiom $A10$ we have that $\forall \psi_G$: \(\chi \wedge [\psi_G \wedge \chi] \varphi \in x\). 
    By the fact that \(\chi \wedge [\psi_G \wedge \chi] \varphi <^{Size}_{[,], [ \! \langle \! \rangle \! ]} [G, \chi] \varphi\) and the Induction Hypothesis, we have \((M^C,x) \models \chi \wedge [\psi_G \wedge \chi] \varphi\) for all $\psi_G$. The latter is equivalent to \((M^C,x) \models [G, \chi] \varphi\) by semantics.
    
    \emph{Case} \(\varphi_0 = [\tau] [G, \chi] \varphi\). Assume that $[\tau] [G, \chi] \varphi \in x$. Note that $[\tau] [G, \chi] \varphi$ is a necessity form. Since $x$ is a maximal consistent theory and, hence, closed under $R5$, we conclude that $\forall \psi_G$: $[\tau] (\chi \wedge [\psi_G \wedge \chi] \varphi) \in x$. Next, by Proposition \ref{prop:measure} and the Induction Hypothesis we have that $(M^C,x) \models [\tau] (\chi \wedge [\psi_G \wedge \chi] \varphi)$ for all $\psi_G$. The latter amounts to the fact that $(M^C,x) \models \tau$ implies $(M^C,x)^\tau \models \chi \wedge [\psi_G \wedge \chi] \varphi$ for all $\psi_G$. By the semantics of \textbf{CoRGAL}, we have that $(M^C,x) \models \tau$ implies $(M^C,x)^\tau \models [G, \chi] \varphi$, which is equivalent to $(M^C,x) \models [\tau] [G, \chi] \varphi$.  
     
    \emph{Case} \(\varphi_0 = [ \! \langle G \rangle \! ] \varphi\). Suppose that \([ \! \langle G \rangle \! ] \varphi \in x\). Since \(x\) is a theory and by axiom $A11$ we have that $\forall \psi_G$: $\langle A \setminus G, \psi_G \rangle \varphi \in x$.
    By the fact that \(\langle A \setminus G, \psi_G \rangle \varphi <^{Size}_{[,], [ \! \langle \! \rangle \! ]} [ \! \langle G \rangle \! ] \varphi\) and the Induction Hypothesis, we have $\forall \psi_G$: \((M^C, x) \models \langle A \setminus G, \psi_G \rangle \varphi\). The latter is $\forall \psi_G, \exists \chi_{A \setminus G}$: \((M^C,x) \models \psi _G \rightarrow \langle \psi_G \wedge \chi_{A \setminus G} \rangle \varphi\) by semantics, which is equivalent to $(M^C,x) \models [\! \langle G \rangle \! ] \varphi$.
    
    \emph{Case} \(\varphi_0 = [\tau] [ \! \langle G \rangle \! ] \varphi\). Assume that $[\tau] [\! \langle G \rangle \!] \varphi \in x$. Note that $[\tau] [\! \langle G \rangle \!] \varphi$ is a necessity form. Since $x$ is a maximal consistent theory and, hence, closed under $R6$, we conclude that $\forall \psi_G$: $[\tau] (\langle A \setminus G, \psi_G \rangle \varphi) \in x$. Next, by Proposition \ref{prop:measure} and the Induction Hypothesis we have that $(M^C,x) \models [\tau] (\langle A \setminus G, \psi_G \rangle \varphi)$ for all $\psi_G$. The latter amounts to the fact that $(M^C,x) \models \tau$ implies $(M^C,x)^\tau \models \langle A \setminus G, \psi_G \rangle \varphi$ for all $\psi_G$. By the semantics of \textbf{CoRGAL}, we have that $(M^C,x) \models \tau$ implies $(M^C,x)^\tau \models [ \! \langle G \rangle \! ] \varphi$, which is equivalent to $(M^C,x) \models [\tau] [ \! \langle G \rangle \! ] \varphi$.  
\end{proof}

Proposition \ref{prop::ph} implies the following fact.

\begin{proposition}
\label{prop::truth}
    Let \(\varphi \in \mathcal{L}_{CoRGAL}\), and \(x\) be a maximal consistent theory. Then \(\varphi \in x\) iff \((M^C,x) \models \varphi\). 
\end{proposition}

Finally, we prove the completeness of \textbf{CoRGAL}.

\begin{proposition}
\label{prop::completeness}
    For all \(\varphi \in \mathcal{L}_{CoRGAL}\), if \(\varphi\) is valid, then \(\varphi \in \mathbf{CoRGAL}\).
\end{proposition}

\begin{proof}
    Towards a contradiction, suppose that \(\varphi\) is valid and \(\varphi \not \in \mathbf{CoRGAL}\). Since \(\mathbf{CoRGAL}\) is a consistent theory, and by Propositions \ref{prop::mct} and \ref{prop::consistency}, we have that \(\mathrm{CoRGAL} + \neg \varphi\) is a consistent theory. Then, by Proposition \ref{prop:lindenbaum}, there exists a maximal consistent theory \(x \supseteq \mathbf{CoRGAL} + \neg \varphi\) such that \(\neg \varphi \in x\). 
    By Proposition \ref{prop::truth}, this means that \((M^C,x) \not \models \varphi\), which contradicts \(\varphi\) being a validity.
\end{proof}

\section{Conclusion} \label{sec:conclusion}

We presented \(\mathbf{CoRGAL}\) and provided a complete axiomatisation for it. 
Validity of \(\langle \! [ G ] \! \rangle \varphi \rightarrow \langle G \rangle [ A \setminus G ] \varphi\) has also been proven. Whether the other direction valid is an open question. Answering it either way, positively, or negatively, will allow us to understand better mutual expressivity of \(\mathbf{CAL}\) and \(\mathbf{GAL}\).
The axiomatisation of \(\mathbf{CoRGAL}\) we presented is infinitary and employs necessity forms. Finding a finitary axiomatisation is yet another open problem. 
An interesting avenue of further research is adding common and distributed knowledge operators to \(\mathbf{CoRGAL}\) in the vein of \cite{agotnesA12}. 
Additionally, since it is known that \(\mathbf{GAL}\), \(\mathbf{CAL}\) \cite{agotnes16}, and hence \(\mathbf{CoRGAL}\), are undecidable, a search for decidable fragments of these logics is another research question. We would also like to investigate applicability of logics with group and coalition announcements to epistemic planning \cite{bolander11}.
Finally, a complete axiomatisation of \(\mathbf{CAL}\) without relativised group announcement operators has not been provided yet, and it is an intriguing direction of further research. 

\section*{Acknowledgements}

We would like to thank three anonymous TARK 2017 reviewers for their insightful suggestions and detailed comments. 

\bibliographystyle{eptcs}
\bibliography{cal}

\appendix

\section{Coalition and Group Annoucement Logic Subsumes Coalition Logic}

\begin{definition}

Axiomatisation of \(\mathbf{CL}\) is as follows:

\(
\begin{array}[h]{ll}
(C0) &\textrm{all instantiation of propositional tautologies},\\

(C1) &\neg \langle \! [ G ] \! \rangle \bot, \\

(C2) &\langle \! [ G ] \! \rangle \top, \\

(C3) &\neg \langle \! [ \emptyset ] \! \rangle \neg \varphi \rightarrow \langle \! [ A ] \! \rangle \varphi,\\

(C4) &\langle \! [ G ] \! \rangle (\varphi \wedge \psi) \rightarrow \langle \! [ G ] \! \rangle \varphi, \\      

(C5) &\langle \! [ G ] \! \rangle \varphi \wedge\langle \! [ H ] \! \rangle \psi \rightarrow \langle \! [ G \cup H ] \! \rangle (\varphi \wedge \psi),\textrm{ if } G \cap H = \emptyset, \\

(R0) &\vdash \varphi, \varphi \rightarrow \psi \Rightarrow \, \vdash \psi,\\

(R1) & \vdash \varphi \leftrightarrow \psi \Rightarrow \, \vdash \langle \! [ G ] \! \rangle \varphi \leftrightarrow \langle \! [ G ] \! \rangle \psi.

\end{array}\)

\end{definition}

\begin{proposition}

\(\mathbf{CoRGAL}\) contains \(\mathbf{CL}\).

\end{proposition}

\begin{proof}

$C0$ and $R0$ are obvious.

$C1$: It holds that $\models \top$, and $\top$ is true in every restriction of a model, i.e. $\models [ \psi ] \top$. In particular, for some model $(M,w)$ and all true formulas $\psi_G$ and $\chi_{A \setminus G}$: $(M,w) \models \langle \psi_G \wedge \chi_{A \setminus G} \rangle \top$. We can relax the requirement of $\psi_G$ being true by adding the formula as an antecedent. Formally, for all (true and false) $\psi_G$ and some (true) $\chi_{A \setminus G}$: $(M,w) \models \psi_G \rightarrow \langle \psi_G \wedge \chi_{A \setminus G} \rangle \top$. The latter is $(M,w) \models [ \! \langle G \rangle \! ] \top$ by the semantics, and this is equivalent to $(M,w) \models \neg \langle \! [ G ] \! \rangle \bot$ by the duality of the coalition announcement operators. 

$C2$: For any pointed model $(M,w)$ and any announcement $\psi_G \wedge \chi_{A \setminus G}$ it holds that $(M,w) \models [\psi_G \wedge \chi_{A \setminus G}] \top$. The latter implies that for some true $\psi_G$ and for all $\chi_{A \setminus G}$: $(M,w) \models \psi_G \wedge [\psi_G \wedge \chi_{A \setminus G}] \top$, which is $(M,w) \models \langle \! [ G ] \! \rangle \top$ by the semantics.

$C3$: Let $\neg \langle \! [ \emptyset ] \! \rangle \neg \varphi$ be true in some arbitrary pointed model $(M,w)$. This is equivalent to $\exists \psi_A$: $(M,w) \models \neg [ \psi_A ] \neg \varphi$, which is $(M,w) \models \langle \! [ A ] \! \rangle \varphi$ by the semantics.

$C4$: Suppose that for some $(M,w)$, $(M,w) \models \langle \! [ G ] \! \rangle (\varphi_1 \wedge \varphi_2)$ holds. By the semantics, $\exists \psi_G, \forall \chi_{A \setminus G}$: $(M,w) \models \psi_G \wedge [\psi_G \wedge \chi_{A \setminus G}] (\varphi_1 \wedge \varphi_2)$. Then, by axiom of \textbf{PAL} $[\psi](\varphi \wedge \chi) \leftrightarrow [\psi] \varphi \wedge [\psi] \chi$, we have $\exists \psi_G, \forall \chi_{A \setminus G}$: $(M,w) \models \psi_G \wedge [\psi_G \wedge \chi_{A \setminus G}] \varphi_1 \wedge  [\psi_G \wedge \chi_{A \setminus G}] \varphi_2$. The latter implies  $\exists \psi_G, \forall \chi_{A \setminus G}$: $(M,w) \models \psi_G \wedge [\psi_G \wedge \chi_{A \setminus G}] \varphi_1$, which is $(M,w) \models \langle \! [ G ] \! \rangle \varphi_1$ by the semantics.

$C5$: Assume that for some $(M,w)$ we have that $(M,w) \models \langle \! [ G ] \! \rangle \varphi_1 \wedge \langle \! [ H ] \! \rangle \varphi_2$. Let us consider the first conjunct $(M,w) \models \langle \! [ G ] \! \rangle \varphi_1$. By the semantics it is equivalent to $\exists \psi_G, \forall \chi_{A \setminus G}$: $(M,w) \models \psi_G \wedge [ \psi_G \wedge \chi_{A \setminus G} ] \varphi_1$. Since $G \cap H = \emptyset$, we can split $\chi_{A \setminus G}$ into $\chi_{H}$ and $\chi_{A \setminus {G \cup H}}$. Thus we have that $\exists \psi_G, \forall \chi_H, \forall \chi_{A \setminus {(G \cup H)}}$: $(M,w) \models \psi_G \wedge [ \psi_G \wedge \chi_{H} \wedge \chi_{A \setminus (G \cup H)} ] \varphi_1$. The same holds for the second conjunct: $\exists \psi_H, \forall \chi_G, \forall \chi_{A \setminus {(G \cup H)}}$: $(M,w) \models \psi_H \wedge [ \psi_H \wedge \chi_{G} \wedge \chi_{A \setminus (G \cup H)} ] \varphi_2$. Since $\chi_H$ ($\chi_G$) quantifies over all formulas known to $H$ ($G$), we can substitute $\chi_H$ ($\chi_G$) with $\psi_H$ ($\psi_G$). Hence we have 
\[\exists \psi_G, \exists \psi_H, \forall \chi_{A \setminus (G \cup H)}:\] \[(M,w) \models \psi_G \wedge \psi_H \wedge [\psi_G \wedge \psi_H \wedge \chi_{A \setminus (G \cup H)}] \varphi_1 \wedge [\psi_G \wedge \psi_H \wedge \chi_{A \setminus G \cup H}] \varphi_2.\] 
By the axiom of \textbf{PAL} $[\psi](\varphi \wedge \chi) \leftrightarrow [\psi] \varphi \wedge [\psi] \chi$, we have that \[\exists \psi_G, \exists \psi_H, \forall \chi_{A \setminus (G \cup H)}: (M,w) \models \psi_G \wedge \psi_H \wedge [\psi_G \wedge \psi_H \wedge \chi_{A \setminus (G \cup H)}] (\varphi_1 \wedge \varphi_2),\] and the latter is equivalent to $(M,w) \models \langle \! [ G \cup H ] \! \rangle (\varphi_1 \wedge \varphi_2)$ by the semantics.

\(R1\): Assume that $\models \varphi \leftrightarrow \psi$. This means that for any pointed model \((M,w)\) the following holds: \((M,w) \models \varphi\) iff \((M,w) \models \psi\) (1). Now suppose that for some pointed model \((M,v)\) it holds that \((M,v) \models \langle \! [ G ] \! \rangle \varphi\). By the semantics, $\exists \psi_G, \forall \chi_{A \setminus G}$: $(M,v) \models \psi_G \wedge [\psi_G \wedge \chi_{A \setminus G}] \varphi$, which is equivalent to the following: $(M,v) \models \psi_G$ and ($(M,v) \models \psi_G \wedge \chi_{A \setminus G}$ implies $(M^{\psi_G \wedge \chi_{A \setminus G}},v) \models \varphi$). By (1) we have that $\exists \psi_G, \forall \chi_{A \setminus G}$: $(M,v) \models \psi_G$ and ($(M,v) \models \psi_G \wedge \chi_{A \setminus G}$ implies $(M^{\psi_G \wedge \chi_{A \setminus G}},v) \models \psi$), which is $(M,v) \models \langle \! [ G ] \! \rangle \psi\) by the semantics. The same argument holds in the other direction.
\end{proof}

\section{Proofs of Propositions}

\begin{proposition}
\label{prop::app1}
    $R3$, $R4$, $R5$, and $R6$ are sound, that is, they preserve validity.
\end{proposition}
\begin{proof}
Proofs of $R3$ and $R4$ are similar, and we present only the proof for $R3$.

($R4$) Assume $\models \varphi$. Since public announcements preserve validity, we have that for any $(M,w)$ and $\psi$, $(M,w) \models [\psi] \varphi$. Since $\psi$ is arbitrary, we have that for all $\psi_G$ and $\chi_{A \setminus G}$ $(M,w) \models [\psi_G \wedge \chi_{A \setminus G}] \varphi$. The latter implies that for some true $\psi_G$ it holds that $(M,w) \models \psi_G \wedge [\psi_G \wedge \chi_{A \setminus G}] \varphi$, which is $(M,w) \models [ \! \langle G \rangle \! ] \varphi$ by the semantics. Since $(M,w)$ was arbitrary, we conclude that $\models [ \! \langle G \rangle \! ] \varphi$.

($R5$) Let $(M,w)$ be an arbitrary pointed model, and let some $\chi$ be given. We proceed by induction on $\eta$.

\textit{Base case}. For all $\psi_G$ we have that $(M,w) \models \chi \wedge [\psi_G \wedge \chi] \varphi$. This is equivalent to $(M,w) \models [G, \chi] \varphi$ by semantics.

\textit{Induction Hypothesis}. Assume that the rule preserves validity for all formulas $\eta(\chi \wedge [\psi_G \wedge \chi] \varphi)$ and all pointed models $(M,w)$.

\textit{Case} $\forall \psi_G$: $\tau \rightarrow \eta(\chi \wedge [\psi_G \wedge \chi] \varphi)$ for some $\tau \in \mathcal{L}_{CoRGAL}$. This means that $(M,w) \models \neg \tau$ or $(M,w) \models \eta(\chi \wedge [\psi_G \wedge \chi] \varphi)$. By the Induction Hypothesis we have that $(M,w) \models \neg \tau$ or $(M,w) \models \eta([ G, \chi ] \varphi)$, which is equivalent to $(M,w) \models \tau \rightarrow \eta([ G, \chi ] \varphi)$.

\textit{Case} $\forall \psi_G$: $K_a \eta(\chi \wedge [\psi_G \wedge \chi] \varphi)$ for some $a \in A$. By semantics we have that for every $v \in W$: $(w,v) \in \sim_a$ implies $(M,v) \models \eta(\chi \wedge [\psi_G \wedge \chi] \varphi)$. By the Induction Hypothesis we conclude that for every $v \in W$: $(w,v) \in \sim_a$ implies $(M,v) \models \eta([G, \chi] \varphi)$, which is equivalent to $(M,w) \models K_a \eta([ G, \chi ] \varphi)$.

\textit{Case} $\forall \psi_G$: $[\tau] \eta(\chi \wedge [\psi_G \wedge \chi] \varphi)$ for some $\tau \in \mathcal{L}_{CoRGAL}$. This means that $(M,w) \models \tau$ implies $(M^\tau,w) \models \eta(\chi \wedge [\psi_G \wedge \chi] \varphi)$. By the Induction Hypothesis we have that $(M,w) \models \tau$ implies $(M^\tau,w) \models \eta([ G, \chi ] \varphi)$, which is equivalent to $(M,w) \models [\tau] \eta([ G, \chi ] \varphi)$.

($R6$) Let $(M,w)$ be an arbitrary pointed model. We proceed by induction on $\eta$.

\textit{Base case}. For all $\psi_G$ we have that $(M,w) \models \langle A \setminus G, \psi_G \rangle \varphi$. This is equivalent to $(M,w) \models [ \! \langle G \rangle \! ] \varphi$ by the alternative semantics using relativised group announcements.

\textit{Induction Hypothesis}. Assume that the rule preserves validity for all formulas $\eta(\langle A \setminus G, \psi_G \rangle \varphi)$ and all pointed models $(M,w)$.

\textit{Case} $\forall \psi_G$: $\tau \rightarrow \eta(\langle A \setminus G, \psi_G \rangle \varphi)$ for some $\tau \in \mathcal{L}_{CoRGAL}$. This means that $(M,w) \models \neg \tau$ or $(M,w) \models \eta(\langle A \setminus G, \psi_G \rangle \varphi)$. By the Induction Hypothesis we have that $(M,w) \models \neg \tau$ or $(M,w) \models \eta([ \! \langle G \rangle \! ] \varphi)$, which is equivalent to $(M,w) \models \tau \rightarrow \eta([ \! \langle G \rangle \! ] \varphi)$.

\textit{Case} $\forall \psi_G$: $K_a \eta(\langle A \setminus G, \psi_G \rangle \varphi)$ for some $a \in A$. By semantics we have that for every $v \in W$: $(w,v) \in \sim_a$ implies $(M,v) \models \eta(\langle A \setminus G, \psi_G \rangle \varphi)$. By the Induction Hypothesis we conclude that for every $v \in W$: $(w,v) \in \sim_a$ implies $(M,v) \models \eta([ \! \langle G \rangle \! ] \varphi)$, which is equivalent to $(M,w) \models K_a \eta([ \! \langle G \rangle \! ] \varphi)$.

\textit{Case} $\forall \psi_G$: $[\tau] \eta(\langle A \setminus G, \psi_G \rangle \varphi)$ for some $\tau \in \mathcal{L}_{CoRGAL}$. This means that $(M,w) \models \tau$ implies $(M^\tau,w) \models \eta(\langle A \setminus G, \psi_G \rangle \varphi)$. By the Induction Hypothesis we have that $(M,w) \models \tau$ implies $(M^\tau,w) \models \eta([ \! \langle G \rangle \! ] \varphi)$, which is equivalent to $(M,w) \models [\tau] \eta([ \! \langle G \rangle \! ] \varphi)$. 
 \end{proof}

\begin{proposition}
\label{prop:oneapp}
$\langle \! [ G ] \! \rangle \varphi \rightarrow \langle \! [ G \cup H ] \! \rangle \varphi $, where $G,H \subseteq A$, is valid.
\end{proposition}

\begin{proof}
Let $(M,w) \models \langle \! [ G ] \! \rangle \varphi$ for some arbitrary $(M,w)$. By the semantics of \textbf{CAL} this is equivalent to \[\exists \psi_G, \forall \chi_{A \setminus G}: (M,w) \models \psi_G \wedge [\psi_G \wedge \chi_{A \setminus G}] \varphi.\]
Let us consider formula $\chi_{A \setminus G}$: since $A \setminus G = A \setminus (G \cup H) \cup H \setminus G$, we can `unpack' the formula into $\chi_{A \setminus (G \cup H)}$ and $\chi_{H \setminus G}$. Hence we have 
\[\exists \psi_G, \forall \chi_{H \setminus G}, \forall \chi_{A \setminus (G \cup H)}: (M,w) \models \psi_G \wedge [\psi_G \wedge \chi_{H \setminus G} \wedge \chi_{A \setminus (G \cup H)}] \varphi.\]
The latter implies 
\[\exists \psi_G, \exists \top_{H \setminus G}, \forall \chi_{A \setminus (G \cup H)}: (M,w) \models \psi_G \wedge \top_{H \setminus G} \wedge [\psi_G \wedge \top_{H \setminus G} \wedge \chi_{A \setminus (G \cup H)}] \varphi,\]
where $\top_{H \setminus G}:= \bigwedge_{a \in H \setminus G} K_a \top$. Combining $\psi_G$ and $\top_{H \setminus G}$ into a single announcement $\psi_{G \cup H}$ by the united coalition $G \cup H$, we conclude that
\[\exists \psi_{G \cup H}, \forall \chi_{A \setminus (G \cup H)}: (M,w) \models \psi_{G \cup H} \wedge [\psi_{G \cup H} \wedge \chi_{A \setminus (G \cup H)}] \varphi.\]
This is equivalent to $(M,w) \models \langle \! [ G \cup H ] \! \rangle \varphi$ by semantics.
\end{proof}

\begin{proposition}
\label{prop:twoapp}
$\langle \! [ G ] \! \rangle \langle \! [ G ] \! \rangle \varphi \rightarrow [ \! \langle A \setminus G \rangle \! ] \varphi$ is valid.
\end{proposition}

\begin{proof}
Suppose that for some $(M,w)$ it holds that $(M,w) \models \langle \! [ G ] \! \rangle \langle \! [ G ] \! \rangle \varphi$. This is equivalent to 
\[\exists \psi_G, \forall \chi_{A \setminus G}, \exists \psi^\prime_G, \forall \chi^\prime_{A \setminus G}: (M,w) \models \psi_G \wedge [\psi_G \wedge \chi_{A \setminus G}](\psi^\prime_G \wedge [\psi_G^\prime \wedge \chi^\prime_{A \setminus G}] \varphi).\] Since $\chi^\prime_{A \setminus G}$ quantifies over all epistemic formulas known to $A \setminus G$, it also quantifies over $\top_{A \setminus G}:= \bigwedge_{a \in A \setminus G} K_a \top$. Hence it is implied that 
\[\exists \psi_G, \forall \chi_{A \setminus G}, \exists \psi^\prime_G: (M,w) \models \psi_G \wedge [\psi_G \wedge \chi_{A \setminus G}](\psi^\prime_G \wedge [\psi_G^\prime \wedge \top_{A \setminus G}] \varphi),\] which is equivalent to 
\[\exists \psi_G, \forall \chi_{A \setminus G}, \exists \psi^\prime_G: (M,w) \models \psi_G \wedge [\psi_G \wedge \chi_{A \setminus G}] \psi^\prime_G \wedge [\psi_G \wedge \chi_{A \setminus G}][\psi_G^\prime] \varphi.\]

Using \textbf{PAL} validity $[\psi] \chi \wedge [\psi][\chi] \varphi \leftrightarrow [\psi] \chi \wedge [\psi] \langle \chi \rangle \varphi$, we get
\[\exists \psi_G, \forall \chi_{A \setminus G}, \exists \psi^\prime_G: (M,w) \models \psi_G \wedge [\psi_G \wedge \chi_{A \setminus G}] \psi^\prime_G \wedge [\psi_G \wedge \chi_{A \setminus G}] \langle \psi_G^\prime \rangle \varphi.\]

Next, we use \textbf{PAL} validity $[\psi] \varphi \leftrightarrow (\psi \rightarrow \langle \psi \rangle \varphi)$: 
\[\exists \psi_G, \forall \chi_{A \setminus G}, \exists \psi^\prime_G: (M,w) \models \psi_G \wedge [\psi_G \wedge \chi_{A \setminus G}] \psi^\prime_G \wedge (\psi_G \wedge \chi_{A \setminus G} \rightarrow \langle \psi_G \wedge \chi_{A \setminus G} \rangle \langle \psi_G^\prime \rangle \varphi).\]

By propositional reasoning the latter implies 
\[\exists \psi_G, \forall \chi_{A \setminus G}, \exists \psi^\prime_G: (M,w) \models \psi_G \wedge (\psi_G \rightarrow (\chi_{A \setminus G} \rightarrow \langle \psi_G \wedge \chi_{A \setminus G} \rangle \langle \psi_G^\prime \rangle \varphi),\]
and this implies 
\[\exists \psi_G, \forall \chi_{A \setminus G}, \exists \psi^\prime_G: (M,w) \models \chi_{A \setminus G} \rightarrow \langle \psi_G \wedge \chi_{A \setminus G} \rangle \langle \psi_G^\prime \rangle \varphi.\]

Finally, by \textbf{PAL} axiom $\langle \psi \rangle \langle \chi \rangle \varphi \leftrightarrow \langle \psi \wedge [\psi] \chi \rangle \varphi$, we have that 
\[\exists \psi_G, \forall \chi_{A \setminus G}, \exists \psi^\prime_G: (M,w) \models \chi_{A \setminus G} \rightarrow \langle \psi_G \wedge \chi_{A \setminus G} \wedge [\psi_G \wedge \chi_{A \setminus G}] \psi_G^\prime \rangle \varphi.\]

We can move $\exists \psi_G$ within the scope of $\forall \chi_{A \setminus G}$, and morph $\psi_G$ and $[\psi_G \wedge \chi_{A \setminus G}] \psi_G^\prime$ into a single announcement by $G$.

The latter is $(M,w) \models [\! \langle A \setminus G \rangle \! ] \varphi$ by the semantics of \textbf{CAL}.
\end{proof}

\begin{lemma}
\label{lemma::2app}
	Let \(\varphi, \psi \in \mathcal{L}_{CoRGAL}\). If \(\varphi \rightarrow \psi\) is a theorem, then \(\eta (\varphi) \rightarrow \eta (\psi)\) is a theorem as well.
\end{lemma}

\begin{proof}
    Assume that $\varphi \rightarrow \psi$ is a theorem. We prove the lemma by induction on $\eta$.
    
    \textit{Base case} $\eta:= \sharp$. Formula $\varphi \rightarrow \psi$ is a theorem by assumption.
    
    \textit{Induction Hypothesis}. Assume that for some $\eta$, $\eta(\varphi) \rightarrow \eta (\psi)$ is a theorem.
    
    \textit{Case} $(\tau \rightarrow \eta (\varphi)) \rightarrow (\tau \rightarrow \eta (\psi))$ for some $\tau \in \mathcal{L}_{CoRGAL}$. Formula $(\eta(\varphi) \rightarrow \eta(\psi)) \rightarrow ((\tau \rightarrow \eta (\varphi)) \rightarrow (\tau \rightarrow \eta (\psi)))$ is a propositional tautology, and, hence, a theorem of \textbf{CoRGAL}. Using the Induction Hypothesis and $R0$, we have that $(\tau \rightarrow \eta (\varphi)) \rightarrow (\tau \rightarrow \eta (\psi))$ is a theorem.
    
\textit{Case} $(K_a \eta (\varphi)) \rightarrow (K_a \eta (\psi))$ for some $a \in A$. Since $\eta(\varphi) \rightarrow \eta(\psi)$ is a theorem by the Induction Hypothesis, $K_a (\eta(\varphi) \rightarrow \eta(\psi))$ is also a theorem by $R1$. Next, $K_a (\eta(\varphi) \rightarrow \eta(\psi)) \rightarrow (K_a \eta (\varphi) \rightarrow K_a \eta (\psi))$ is an instance of $A1$, and, hence, a theorem. Finally, using $R0$ we have that $K_a \eta (\varphi) \rightarrow K_a \eta (\psi)$ is a theorem.
        
\textit{Case} $([\tau] \eta (\varphi)) \rightarrow ([\tau] \eta (\psi))$ for some $\tau \in \mathcal{L}_{CoRGAL}$. Formula $[\tau](\eta(\varphi) \rightarrow \eta(\psi)) \rightarrow ([\tau]\eta(\varphi) \rightarrow [\tau]\eta(\psi))$ is a theorem of \textbf{PAL} (see \cite[Chapter 4]{del}), and hence of \textbf{CoRGAL}. Using the Induction Hypothesis and $R0$ we conclude that $[\tau]\eta(\varphi) \rightarrow [\tau]\eta(\psi)$ is also a theorem of \textbf{CoRGAL}. 
\end{proof}

\begin{proposition}
\label{prop::mctapp}
    Let \(x\) be a theory, \(\varphi, \psi \in \mathcal{L}_{CoRGAL}\), and \(a \in A\). The following are theories: \(x + \varphi = \{\psi: \varphi \rightarrow \psi \in x\}, K_a x = \{\varphi: K_a \varphi \in x\}\), and \([\varphi]x = \{\psi: [\varphi]\psi \in x\}\).
\end{proposition}

\begin{proof}
The proof is an extension of the one from \cite{balbiani08}. We show that corresponding sets are closed under $R5$ and $R6$. 

\textit{Case} $x + \varphi$.
    Suppose that \(\eta (\chi \wedge [\psi_G \wedge \chi]\tau) \in x + \varphi\) for some given $\chi$, for all $\psi_G$, and for some $\tau \in \mathcal{L}_{CoRGAL}$. 
    This means that \(\varphi \rightarrow \eta (\chi \wedge [\psi_G \wedge \chi]\tau) \in x\) for all $\psi_G$. 
    Since \(\varphi \rightarrow \eta (\chi \wedge [\psi_G \wedge \chi]\tau)\) is a necessity form, and \(x\) is closed under \(R5\) (by Definition \ref{def:theory}), we infer that \(\varphi \rightarrow \eta ([G, \chi] \tau) \in x\), and, consequently, \(\eta ([G, \chi] \tau) \in x + \varphi\). 
    So, \(x + \varphi\) is closed under \(R5\).
    
    Now, let $\forall \psi_G$: $\eta(\langle A \setminus G, \psi_G \rangle \tau) \in x + \varphi$. By the definition of $x+\varphi$ this means that $\varphi \rightarrow \eta(\langle A \setminus G, \psi_G \rangle \tau) \in x$ for all $\psi_G$. Since $\varphi \rightarrow \eta(\langle A \setminus G, \psi_G \rangle \tau$ is a necessity form and $x$ is closed under $R6$, we infer that \(\varphi \rightarrow \eta( [ \! \langle G \rangle \! ] \tau) \in x\), and, consequently, \(\eta( [ \! \langle G \rangle \! ] \tau) \in x + \varphi\). So, \(x + \varphi\) is closed under \(R6\).
    
    \textit{Case} $K_a x$.
    Suppose that \(\eta (\chi \wedge [\psi_G \wedge \chi]\tau) \in K_a x\) for some given $\chi$, for all $\psi_G$, and for some $\tau \in \mathcal{L}_{CoRGAL}$. 
    This means that \(K_a \eta (\chi \wedge [\psi_G \wedge \chi]\tau) \in x\) for all $\psi_G$. 
    Since \(K_a \eta (\chi \wedge [\psi_G \wedge \chi]\tau)\) is a necessity form, and \(x\) is closed under \(R5\) (by Definition \ref{def:theory}), we infer that \(K_a \eta ([G, \chi] \tau) \in x\), and, consequently, \(\eta ([G, \chi] \tau) \in K_a x\). 
    So, \(K_a x\) is closed under \(R5\).
    
    Now, let $\forall \psi_G$: $\eta(\langle A \setminus G, \psi_G \rangle \tau) \in K_a x$. By the definition of $K_a x$ this means that $K_a \eta(\langle A \setminus G, \psi_G \rangle \tau) \in x$ for all $\psi_G$. Since $K_a \eta(\langle A \setminus G, \psi_G \rangle \tau$ is a necessity form and $x$ is closed under $R6$, we infer that \(K_a \eta( [ \! \langle G \rangle \! ] \tau) \in x\), and, consequently, \(\eta( [ \! \langle G \rangle \! ] \tau) \in K_a x\). So, \(K_a x\) is closed under \(R6\).
    
        \textit{Case} $[\varphi] x$.
 Finally, suppose that \(\eta (\chi \wedge [\psi_G \wedge \chi]\tau) \in [\varphi] x\) for some given $\chi$, for all $\psi_G$, and for some $\tau \in \mathcal{L}_{CoRGAL}$. 
    This means that \([\varphi] \eta (\chi \wedge [\psi_G \wedge \chi]\tau) \in x\) for all $\psi_G$. 
    Since \([\varphi] \eta (\chi \wedge [\psi_G \wedge \chi]\tau)\) is a necessity form, and \(x\) is closed under \(R5\) (by Definition \ref{def:theory}), we infer that \([\varphi] \eta ([G, \chi] \tau) \in x\), and, consequently, \(\eta ([G, \chi] \tau) \in [\varphi] x\). 
    So, \([\varphi] x\) is closed under \(R5\).
    
    Now, let $\forall \psi_G$: $\eta(\langle A \setminus G, \psi_G \rangle \tau) \in [\varphi] x$. By the definition of $[\varphi] x$ this means that $[\varphi] \eta(\langle A \setminus G, \psi_G \rangle \tau) \in x$ for all $\psi_G$. Since $[\varphi] \eta(\langle A \setminus G, \psi_G \rangle \tau$ is a necessity form and $x$ is closed under $R6$, we infer that \([\varphi] \eta( [ \! \langle G \rangle \! ] \tau) \in x\), and, consequently, \(\eta( [ \! \langle G \rangle \! ] \tau) \in [\varphi] x\). So, \([\varphi] x\) is closed under \(R6\).
\end{proof}

\begin{proposition}
\label{prop::consistencyapp}
Let \(\varphi \in \mathcal{L}_{CoRGAL}\). Then \(\mathbf{CoRGAL} + \varphi\) is consistent iff \(\neg \varphi \not \in \mathbf{CoRGAL}\).
\end{proposition}

\begin{proof}

  \textit{From left to right}. Suppose to the contrary that \(\mathbf{CoRGAL} + \varphi\) is consistent and \(\neg \varphi \in \mathbf{CoRGAL}\). Then having both \(\varphi\) and \(\neg \varphi\) means that \(\bot \in \mathbf{CoRGAL}+\varphi\), which contradicts to \(\mathbf{CoRGAL}+\varphi\) being consistent. 
  
  \textit{From right to left}. Let us consider the contrapositive: if \(\mathbf{CoRGAL} + \varphi\) is inconsistent, then \(\neg \varphi \in \mathbf{CoRGAL}\). Since \(\mathbf{CoRGAL} + \varphi\) is inconsistent, \(\bot \in \mathbf{CoRGAL} + \varphi\), or, by Proposition \ref{prop::mct}, \(\varphi \rightarrow \bot \in \mathbf{CoRGAL}\). By consistency of \(\mathbf{CoRGAL}\) and propositional reasoning, we have that \(\neg \varphi \in \mathbf{CoRGAL}\).
\end{proof} 

 \begin{proposition}
  \label{prop:measureapp}
  Let $\psi_G$, $G \subseteq A$, and $\chi, \varphi, \tau \in \mathbf{CoRGAL}$.
\begin{enumerate}
\item $\chi \wedge [\psi_G \wedge \chi]\varphi <^{Size}_{[,], [ \! \langle \! \rangle \! ]} [G, \chi]\varphi$,

\item $[\tau] (\chi \wedge [\psi_G \wedge \chi]\varphi) <^{Size}_{[,], [ \! \langle \! \rangle \! ]} [\tau][G, \chi]\varphi$,

\item $\langle A \setminus G, \psi_G \rangle \varphi <^{Size}_{[,], [ \! \langle \! \rangle \! ]} [ \! \langle G \rangle \! ] \varphi$,

\item $[\tau] \langle A \setminus G, \psi_G \rangle \varphi <^{Size}_{[,], [ \! \langle \! \rangle \! ]} [\tau] [ \! \langle G \rangle \! ] \varphi$.
\end{enumerate}
\end{proposition}

\begin{proof}

\begin{enumerate}
\item Note that $[\! \langle \! \rangle \! ]$-depth for both sides of the inequality is the same and equals $d_{[ \! \langle \! \rangle \! ]} (\chi) + d_{[ \! \langle \! \rangle \! ]} (\varphi)$. In particular,  $d_{[ \! \langle \! \rangle \! ]} (\chi \wedge [\psi_G \wedge \chi]\varphi) =$ $\textrm{max} \{d_{[ \! \langle \! \rangle \! ]} (\chi),$ $d_{[ \! \langle \! \rangle \! ]} ([\psi_G \wedge \chi]\varphi)\} =$ $d_{[ \! \langle \! \rangle \! ]} ([\psi_G \wedge \chi]\varphi)=$ $d_{[ \! \langle \! \rangle \! ]} (\psi_G \wedge \chi) + d_{[ \! \langle \! \rangle \! ]} (\varphi) = $ $d_{[ \! \langle \! \rangle \! ]} (\chi)+$ $d_{[ \! \langle \! \rangle \! ]} (\varphi)$. Depth of the right-hand side formula is $d_{[ \! \langle \! \rangle \! ]} ([G, \chi]\varphi) = $ $d_{[ \! \langle \! \rangle \! ]} (\chi) +$ $d_{[ \! \langle \! \rangle \! ]} (\varphi)$. However, $[,]$-depth is different. Indeed, $d_{[,]} ([\psi_G \wedge \chi] \varphi) = $ $d_{[,]} (\psi_G \wedge \chi) + $ $d_{[,]} (\varphi) = $ $d_{[,]}(\chi) + d_{[,]} (\varphi)$. For the right-hand side formula we have that $d_{[,]} ([G,\chi] \varphi) = $ $d_{[,]} (\chi) + $ $d_{[,]} (\varphi) + 1$. Hence, $\chi \wedge [\psi_G \wedge \chi]\varphi <^{Size}_{[,], [ \! \langle \! \rangle \! ]} [G, \chi]\varphi$. 

\item On the left-hand side we have $d_{[ \! \langle \! \rangle \! ]} ([\tau] (\chi \wedge [\psi_G \wedge \chi]\varphi)) = $ $d_{[ \! \langle \! \rangle \! ]} (\tau) + d_{[ \! \langle \! \rangle \! ]} (\chi \wedge [\psi_G \wedge \chi]\varphi) = $ $d_{[ \! \langle \! \rangle \! ]} (\tau) + d_{[ \! \langle \! \rangle \! ]} (\chi) + $ $d_{[ \! \langle \! \rangle \! ]} (\varphi)$. We have the same $[ \! \langle \! \rangle \! ]$-depth of the right-hand side:  $d_{ [ \! \langle \! \rangle \! ]} ([\tau] [G, \chi] \varphi) = $ $d_{ [ \! \langle \! \rangle \! ]} (\tau) + d_{ [ \! \langle \! \rangle \! ]} ([G, \chi] \varphi) = $ $d_{ [ \! \langle \! \rangle \! ]} (\tau) +$ $d_{[ \! \langle \! \rangle \! ]} (\chi) + $ $d_{ [ \! \langle \! \rangle \! ]} (\varphi)$. However, $[,]$-depth is different: $d_{[,]} (\tau) + d_{[,]} (\chi) + d_{[,]} (\varphi)$ and $d_{[,]} (\tau) + d_{[,]} (\chi) + d_{[,]} (\varphi) + 1$ correspondingly (see the previous case). Hence, $[\tau] (\chi \wedge [\psi_G \wedge \chi]\varphi) <^{Size}_{[,], [ \! \langle \! \rangle \! ]} [\tau][G, \chi]\varphi$. 

\item On the left-hand side we have that  $d_{[ \! \langle \! \rangle \! ]} (\langle A \setminus G, \psi_G \rangle \varphi) =$ $d_{[ \! \langle \! \rangle \! ]} (\varphi)$, and on the right-hand side the depth is $d_{[ \! \langle \! \rangle \! ]} [ \! \langle G \rangle \! ] \varphi = $ $d_{[ \! \langle \! \rangle \! ]} (\varphi) + 1$. Hence, $\langle A \setminus G, \psi_G \rangle \varphi <^{Size}_{[,], [ \! \langle \! \rangle \! ]} [ \! \langle G \rangle \! ] \varphi$.

\item Again, according to the definition of $[ \! \langle \! \rangle \! ]$-depth, $d_{[ \! \langle \! \rangle \! ]} ([\tau] \langle A \setminus G, \psi_G \rangle \varphi) = $ $d_{[ \! \langle \! \rangle \! ]} (\tau) + d_{[ \! \langle \! \rangle \! ]} (\langle A \setminus G, \psi_G \rangle \varphi) = $ $d_{[ \! \langle \! \rangle \! ]} (\tau) + d_{[ \! \langle \! \rangle \! ]} (\varphi)$, whereas $d_{[ \! \langle \! \rangle \! ]} ([\tau] [ \! \langle G \rangle \! ] \varphi) = $ $d_{[ \! \langle \! \rangle \! ]} (\tau) + d_{[ \! \langle \! \rangle \! ]} ([ \! \langle G \rangle \! ] \varphi) = $ $d_{[ \! \langle \! \rangle \! ]} (\tau) + d_{[ \! \langle \! \rangle \! ]} (\varphi) + 1$. Thus, $[\tau] \langle A \setminus G, \psi_G \rangle \varphi <^{Size}_{[,], [ \! \langle \! \rangle \! ]} [\tau] [ \! \langle G \rangle \! ] \varphi$.   
\end{enumerate}
\end{proof}

\end{document}